\documentclass[12pt]{article}

\usepackage[hmargin=1.3in,vmargin=1.3in]{geometry}
\usepackage{bbding}
\usepackage{mathrsfs}
\usepackage{}
\usepackage{amsfonts}
\usepackage{multirow}
\usepackage{amsfonts,amssymb,amsmath,amsthm,bm}

\newtheorem{theorem}{Theorem}
\newtheorem{example}{Example}

\newtheorem{remark}{Remark}
\newtheorem{corollary}{Corollary}
\newtheorem{lemma}{Lemma}

\begin{document}

\title{Some New Constructions of Quantum MDS Codes}
\author{\small Weijun Fang\thanks{Corresponding Author} $^{,1,2}$  \ Fang-Wei Fu$^{1}$ \\
\small $^1$  Chern Institute of Mathematics and LPMC, Nankai University, Tianjin,  China\\
\small $^2$ Shenzhen International Graduate School, Tsinghua University, Shenzhen,  China \\
\small Email: nankaifwj@163.com, fwfu@nankai.edu.cn\\
}
\date{}
\maketitle
\thispagestyle{empty}
\begin{abstract}
It is an important task to construct quantum maximum-distance-separable (MDS) codes with good parameters. In the present paper, we provide six new classes of $q$-ary quantum MDS codes by using generalized Reed-Solomon (GRS) codes and Hermitian construction. The minimum distances of our quantum MDS codes can be larger than $\frac{q}{2}+1$. Three of these six classes of quantum MDS codes have longer lengths than the ones constructed in \cite{ZG17} and \cite{SYZ17}, hence some of their results can be easily derived from ours via the propagation rule. Moreover, some known quantum MDS codes of specific lengths can be seen as special cases of ours and the minimum distances of some known quantum MDS codes are also improved as well.
\end{abstract}


\maketitle

\section{Introduction}

Quantum error-correcting codes play an important role in quantum
computing and quantum communication. Just as in classical coding theory, one central theme in quantum error-correction is the construction
of quantum codes that have good parameters.  In \cite{CRSS98}, Calderbank \emph{et al.} presented an effective method to construct nice
quantum codes by using some mathematical techniques which made it possible to construct quantum codes
from classical codes over $\mathbb{F}_{2}$ or $\mathbb{F}_{4}$. Rains \cite{R99},  Ashikhmin and  Knill \cite{AK01} then generalized their results to the nonbinary cases. In particular, one can construct quantum codes via classical codes with Euclidean or Hermitian self-orthogonality properties.

Let $q$ be a prime power. A $q$-ary quantum code is just a vector subspace of the Hilbert space $(\mathbb{C}^{q})^{\bigotimes n}\cong\mathbb{C}^{q^{n}}$, where $\mathbb{C}$ is the field of complex numbers and $n$ is called the length of the quantum code. We use $((n, K, d))_{q}$ or $[[n, k, d]]_{q}$ to denote a $q$-ary quantum code of length
$n$, dimension $K$ and minimum distance $d$, where $k=\log_{q}K$. An $[[n, k, d]]_{q}$ quantum code can detect up to $d - 1$ quantum errors and correct up
to $\lfloor\frac{d-1}{2}\rfloor$ quantum errors. Thus for fixed $n$ and $k$, it is desirable to construct $[[n, k, d]]_{q}$-quantum codes with minimum distance $d$ as large as possible. However, similar to
the classical Singleton bound, the parameters of an $[[n, k, d]]_{q}$ quantum code have to satisfy the quantum Singleton bound:
\begin{lemma} (\cite{R99,AK01,KKKS06} Quantum Singleton Bound)\label{lem1.1}
For any $[[n, k, d]]_{q}$ quantum code, we have
\[2d \leq n - k + 2.\]
\end{lemma}
 A quantum code achieving this quantum Singleton bound is called a \emph{quantum maximum-distance-separable
(MDS) code}. Just as in the classical case, it is desirable to find more constructions of quantum MDS codes.

In 2001, Ashikhmin and Knill \cite{AK01} gave the following useful theorem
for constructing quantum stabilizer codes from classical codes.

\begin{theorem} (Hermitian Construction)\label{thm1.2}
  If there exists an $[n, k, d]_{q^{2}}$-linear code $C$ with $C^{\perp_{H}} \subseteq C$, where $C^{\perp_{H}}$ is the Hermitian dual code of $C$, then there exists an $[[n, 2k-n, \geq d]]_{q}$-quantum code.
\end{theorem}

Note that the Hermitian dual code of an MDS code is still an MDS code. So we replace the code $C$ by its Hermitian dual $C^{\perp_{H}}$ in Theorem \ref{thm1.2} and obtain the following corollary for the quantum MDS codes.

\begin{corollary} (Hermitian Construction for Quantum MDS Codes)\label{cor1.3}
If there exists an $[n, k, n-k+1]_{q^{2}}$-MDS code $C$ with $C \subseteq C^{\perp_{H}}$, then there exists an $[[n, n-2k, k+1]]_{q}$-quantum MDS code.
\end{corollary}

Given a quantum MDS code, we can obtain a new quantum code with smaller length and minimum distance by the following lemma.
\begin{lemma} (\cite{GR15} Propagation Rule)\label{lem1.4}
If there exists an $[[n, n-2d+2, d]]_{q}$-quantum MDS code, then there exists an $[[n-1, n-2d+3, d-1]]_{q}$-quantum MDS code.
\end{lemma}

In the past decade, a lot of research work has been done for construction of quantum MDS codes and several new families of quantum MDS codes have been found by employing different methods. If the classical MDS conjecture is true, then there are no $q$-ary quantum MDS codes of length $n$ exceeding $q^{2}+1$ except when $q$ is even and $d=4$ or $d=q^{2}$ in which case $n \leq q^{2}+2$ (see \cite{KKKS06}). Quantum MDS codes of length up to $q + 1$ have been constructed for all possible dimensions through classical Euclidean self-orthogonal codes (see \cite{RGB04, GBR04, JX12}). Since the constraint of Euclidean self-orthogonality, the minimum distance of these quantum MDS codes is less than or equal to $\frac{q}{2}+1$. Thus Hermitian self-orthogonal codes are applied to construct quantum MDS codes with larger minimum distance. Some quantum MDS codes of length $n$ with specific values $n=q^{2}+1, q^{2}, \frac{q^{2}+1}{2}$ and minimum distance $d >q/2+1$ are obtained  (see \cite{GBR04, G11, KZ12}). Due to their elegant algebraic structures, constacyclic codes, pseudo-cyclic codes and generalized Reed-Muller codes are also used to construct some quantum MDS codes of length $n$ with $q+1 < n \leq q^{2}+1$ and relatively large minimum distance (see \cite{KZ12, AKS07, KZL14,WZ15,ZG15,ZC14,CLZ15, LMG16, SK05}). In \cite{LXW08}, Li \emph{et al.} first presented a unified framework for constructing quantum MDS codes by employing the classical generalized Reed-Solomon (GRS) codes.  Jin \emph{et al.} \cite{JLLX10}, Jin and Xing \cite{JX12,JX14} generalized and developed the method in \cite{LXW08}, and constructed several new families of quantum MDS codes with flexible parameters. Since then, GRS codes have been widely applied for constructing quantum MDS codes with minimum distance larger than $\frac{q}{2}+1$ in recent years (see \cite{JKW17,ZG17,SYZ17,FF18}).

In this paper, we will construct some new quantum MDS codes with relatively large minimum distance through classical Hermitian self-orthogonal GRS codes. The key point of constructing Hermitian self-orthogonal GRS codes is to find suitable evaluation points $a_{1}, a_{2}, \ldots, a_{n} \in \mathbb{F}_{q^{2}}$, such that a certain system
of homogenous equations over $\mathbb{F}_{q^{2}}$ related to these evaluation points has solutions over $\mathbb{F}_{q}^{*}$ (see Lemma \ref{lem2.1} and Remark \ref{rem2.2}). In \cite{JX14}, Jin and Xing first chose a class of multiplicative subgroups of $\mathbb{F}^{*}_{q^{2}}$ as the evaluation points to construct Hermitian self-orthogonal GRS codes. In \cite{ZG17, SYZ17}, the authors generalized the method of \cite{JX14} and considered some multiplicative subgroups of $\mathbb{F}^{*}_{q^{2}}$ and their cosets as the evaluation points. In the present paper, we consider some multiplicative subgroups of $\mathbb{F}^{*}_{q^{2}}$ and their cosets with more general parameters. Moreover, we add the zero element into them so that we can provide more constructions of new quantum MDS codes with longer lengths. Consequently, some known results can be easily derived from ours by the propagation rule of Lemma \ref{lem1.4}. More precisely, we provide some $[[n, n - 2k, k + 1]]_{q}$-quantum MDS codes with the
following parameters:
\begin{description}
  \item[\textnormal{(i)}] $n=1+r\frac{q^{2}-1}{s}$, and $1 \leq k \leq r\frac{q-1}{s}$, where $s \mid (q-1)$ and $1 \leq r \leq s$ (See Theorem \ref{thm3.2});
  \item[\textnormal{(ii)}] $n=1+r\frac{q^{2}-1}{2s+1}$, and $1 \leq k \leq (s+1)\frac{q+1}{2s+1}-1$, where $q > 2$, $(2s+1) \mid (q+1)$ and $1 \leq r \leq 2s+1$ (See Theorem \ref{thm4.3} (i));
  \item[\textnormal{(iii)}] $n=1+(2t+1)\frac{q^{2}-1}{2s+1}$, and $1 \leq k \leq (s+t+1)\frac{q+1}{2s+1}-1$, where $q > 2$, $(2s+1) \mid (q+1)$ and $0 \leq t \leq s-1$ (See Theorem \ref{thm4.3} (ii));
  \item[\textnormal{(iv)}] $n=1+r\frac{q^{2}-1}{2s}$, and $1 \leq k \leq (s+1)\frac{q+1}{2s}-1$, where $2s \mid (q+1)$ and $2 \leq r \leq 2s$ (See Theorem \ref{thm5.3} (i));
  \item[\textnormal{(v)}] $n=1+(2t+2)\frac{q^{2}-1}{2s}$, and $1 \leq k \leq (s+t+1)\frac{q+1}{2s}-1$, where $2s \mid (q+1)$ and $0 \leq t \leq s-2$ (See Theorem \ref{thm5.3} (ii));
  \item[\textnormal{(vi)}] $n=(2t+1)\frac{q^{2}-1}{2s}$, and $1 \leq k \leq (s+t)\frac{q+1}{2s}-2$, where $2s \mid (q+1)$ and $1 \leq t \leq s-1$ (See Theorem \ref{thm6.3}).
\end{description}

We make some remarks as follows:

\begin{enumerate}
  \item The minimum distances of quantum MDS codes of cases (i)-(vi) can be larger than or equal to $\frac{q}{2}+1$ (for case (i), we let $\frac{r}{s}\geq \frac{q}{2(q-1)}$);
  \item  Applying the propagation rule (see Lemma \ref{lem1.4}) for cases (i), (iv) and (v), we obtain the results presented in \cite[Theorem 4.12]{SYZ17}, \cite[Theorem 4.2]{ZG17} and \cite[Theorem 4.8]{SYZ17}, respectively;
  \item  The case (ii) extends the result of \cite[Theorem 3.2 (i)]{JKW17}  where a stricter condition $\textnormal{gcd}(r,q) = 1$  is required;
  \item  When $r=2t+1$ (resp. $r=2t+2$) and $t >0$, the codes from case (iii) (resp. (v)) have the same length but larger minimum distance than that of case (ii) (resp. (iv));
  \item When $t \geq 2$, the quantum MDS codes from case (vi) have larger minimum distance than that of \cite[Theorem 4.2]{ZG17}.
\end{enumerate}

We list some examples of $[[n, n-2k, k+1]]_{q}$-quantum MDS codes from our constructions as follows.
\begin{description}
  \item[(i)] $5 \mid (q-1)$, $n=1+\frac{4}{5}(q^{2}-1)$, $1 \leq k \leq \frac{4}{5}(q-1)$;
  \item[(ii)] $5 \mid (q+1)$, $n=1+\frac{2}{5}(q^{2}-1)$, $1 \leq k \leq \frac{3}{5}(q+1)-1$;
  \item[(iii)] $7 \mid (q+1)$, $n=1+\frac{5}{7}(q^{2}-1)$, $1 \leq k \leq \frac{6}{7}(q+1)-1$;
  \item[(iv)] $4 \mid (q+1)$, $n=1+\frac{3}{4}(q^{2}-1)$, $1 \leq k \leq \frac{3}{4}(q+1)-1$;
  \item[(v)] $6 \mid (q+1)$, $n=1+\frac{2}{3}(q^{2}-1)$, $1 \leq k \leq \frac{5}{6}(q+1)-1$;
  \item[(vi)] $8 \mid (q+1)$, $n=\frac{7}{8}(q^{2}-1)$, $1 \leq k \leq \frac{7}{8}(q+1)-2$.
\end{description}
To the best of our knowledge, all the above quantum MDS codes are new.

The rest of this paper is organized as follows. In Section 2, we recall some basic results about Hermitian self-orthogonality and generalized Reed-Solomon codes. In Sections 3, 4, 5 and 6, we present six new classes of quantum MDS codes from generalized Reed-Solomon codes. We conclude this paper in Section 7.

\section{Preliminaries}
\label{sec:1}
In this section, we briefly review some basic results about Hermitian self-orthogonality and
generalized Reed-Solomon (GRS for short) codes. In addition, some technical lemmas for our constructions are also presented.

Let $q$  be a prime power. Let $\mathbb{F}_{q}$ be the finite field with $q$ elements and $\mathbb{F}_{q}^{*}$ be the multiplicative
group of nonzero elements of $\mathbb{F}_{q}$. A $q$-ary $[n, k, d]$-linear code is just a vector subspace of $\mathbb{F}_{q}^{n}$ with dimension $k$ and minimum Hamming distance $d$, and $n$ is called the length of the code. It is well known that $n$, $k$ and $d$ have to satisfy the Singleton bound: $d \leq n-k+1$. A code achieving the Singleton bound is called a \emph{maximum distance separable} (MDS) code.

Throughout this paper, we denote the all zero vector by $\textbf{0}$. For a vector $\textbf{c}=(c_{1}, \ldots, c_{n}) \in  \mathbb{F}_{q^{2}}^{n}$, we denote by $\textbf{c}^{i}$ the vector $(c_{1}^{i}, \ldots, c_{n}^{i})$. And $0^{0}$ is set to be 1. For any two vectors $\textbf{x}=(x_{1}, \ldots, x_{n}) \in \mathbb{F}^{n}_{q^{2}}$ and $\textbf{y}=(y_{1}, \ldots, y_{n}) \in \mathbb{F}^{n}_{q^{2}}$, the usual Euclidean product of $\textbf{x}$ and $\textbf{y}$  is defined as $\langle \textbf{x} , \textbf{y} \rangle\triangleq\sum_{i=1}^{n}x_{i}y_{i}$. For a linear code $C$ of length $n$ over $\mathbb{F}_{q^{2}}$,
the Euclidean dual code of $C$ is defined as
\[C^{\perp} := \{\textbf{x} \in \mathbb{F}_{q^{2}}^{n} : \langle \textbf{x}, \textbf{c} \rangle =0 ,\textnormal{ for all } \textbf{c} \in C \},\]
and the Hermitian dual code of $C$ is defined as
\[C^{\perp_{H}} := \{\textbf{x} \in \mathbb{F}_{q^{2}}^{n} : \langle \textbf{x}, \textbf{c}^{q} \rangle =0 ,\textnormal{ for all } \textbf{c} \in C \}.\]
The code $C$ is called Hermitian self-orthogonal if $C \subseteq C^{\perp_{H}}$.
It is easy to show that $C^{\perp_{H}} = (C^{(q)})^{\perp}$, where $C^{(q)} = \{\textbf{c}^{q} : \textbf{c} \in C \}$. For a matrix $A=(a_{ij})$ over $\mathbb{F}_{q^{2}}$,  we denote by $A^{(q)}$ the matrix $(a_{ij}^{q}).$ Let $C$ be a linear code over $\mathbb{F}_{q^{2}}$ with a generator matrix $G$, then $G^{(q)}$ is a generator matrix of $C^{(q)}$ hence a parity-check matrix of $C^{\perp_{H}}$.

Choose $n$ distinct elements $a_{1}, \ldots, a_{n}$ of $\mathbb{F}_{q^{2}}$ and $n$ nonzero elements $v_{1}, \ldots, v_{n}$ of $\mathbb{F}_{q^{2}}^{*}$. Put $\textbf{a}= (a_{1}, \ldots, a_{n})$ and $\textbf{v}=(v_{1}, \ldots, v_{n})$. Then the generalized
Reed-Solomon code over $\mathbb{F}_{q^{2}}$ associated to $\textbf{a}$ and $\textbf{v}$ is defined as follows.
\begin{eqnarray*}
  GRS_{k}(\textbf{a}, \textbf{v}) &\triangleq& \{(v_{1}f(a_{1}), \ldots, v_{n}f(a_{n})) \\
    && : f(x) \in \mathbb{F}_{q^{2}}[x], \textnormal{ and deg}(f(x)) \leq k-1 \}.
\end{eqnarray*}
It is well known that the code $GRS_{k}(\textbf{a}, \textbf{v})$ is a $q^{2}$-ary $[n, k, n - k + 1]$-MDS code.
A generator matrix of $GRS_{k}(\textbf{a}, \textbf{v})$ is given by
\begin{equation*}
  G_{k}(\textbf{a}, \textbf{v})=\left(
      \begin{array}{cccc}
        v_{1} & v_{2} & \cdots & v_{n} \\
        v_{1}a_{1} & v_{2}a_{2} & \cdots & v_{n}a_{n} \\
        \vdots & \vdots & \ddots & \vdots \\
        v_{1}a_{1}^{k-1} & v_{2}a_{2}^{k-1} & \cdots & v_{n}a_{n}^{k-1} \\
      \end{array}
    \right).
\end{equation*}

From the above discussion, we can easily obtain the following useful lemma, which was also given in \cite{ZG17,SYZ17,JX14}.
\begin{lemma} (\cite{ZG17,SYZ17, JX14})\label{lem2.1}
Let $a_{1}, \ldots, a_{n}$ be $n$ pairwise distinct elements of $\mathbb{F}_{q^{2}}$ and let $v_{1}, \ldots, v_{n}$ be $n$ nonzero elements of $\mathbb{F}_{q^{2}}^{*}$. Put $\textbf{a}= (a_{1}, \ldots, a_{n})$ and $\textbf{v}=(v_{1}, \ldots, v_{n})$. Then the GRS code $GRS_{k}(\textbf{a}, \textbf{v})$ is Hermitian self-orthogonal if and only if $\langle \textbf{a}^{qi+j}, \textbf{v}^{q+1} \rangle = 0$, for all $0 \leq i, j \leq k-1$.
\end{lemma}

\begin{remark}\label{rem2.2}
If we set $\textbf{u}=(u_{1},u_{2},\ldots, u_{n}):=\textbf{v}^{q+1}$, then $\textbf{u} \in (\mathbb{F}^{*}_{q})^{n}$. Thus from Lemma \ref{lem2.1}, to construct a Hermitian self-orthogonal MDS code, it is sufficient to make sure that the system of homogenous equations $\langle \textbf{a}^{qi+j}, \textbf{u} \rangle = 0$ (for all $0 \leq i, j \leq k-1$) over $\mathbb{F}_{q^{2}}$ has a solution $\textbf{u} \in (\mathbb{F}^{*}_{q})^{n}$.
\end{remark}

Before giving our constructions, we need two technical lemmas. The first lemma provides a sufficient condition under which a certain system of homogenous equations over $\mathbb{F}_{q^{2}}$ has solutions over $\mathbb{F}^{*}_{q}$.
\begin{lemma}\label{lem2.3}
Suppose $r >0$. Let $A$ be an $r \times (r+1)$ matrix over $\mathbb{F}_{q^{2}}$ and satisfy the following two properties: 1) any $r$ columns of $A$ are linearly independent; 2) $A^{(q)}$ is row equivalent to $A$. Then the following system of homogenous equations
$A\textbf{u}^{T}=\textbf{0}^{T}$
has a solution $\textbf{u}=(u_{0}, u_{1}, \ldots, u_{r}) \in (\mathbb{F}^{*}_{q})^{r+1}$.
\end{lemma}

\begin{proof}
From Property 1), the rank of $A$ is equal to $r$. By  Property 2) and \cite[Theorem 2.2]{JX14},  the system of homogenous equations
$A\textbf{u}^{T}=\textbf{0}^{T}$ has a nonzero solution $\textbf{u} \in (\mathbb{F}_{q})^{r+1}$. Let $C$ be the linear code over $\mathbb{F}_{q^{2}}$ with generator matrix $A$. Then $C$ is an $[r+1, r, 2]$-MDS code from Property 1) and $\textbf{u}$ is a nonzero codeword of $C^{\perp}$. Note that $C^{\perp}$ is an $[r+1, 1, r+1]$-MDS code, thus $\textbf{u} \in (\mathbb{F}^{*}_{q^{2}})^{r+1}$, hence $\textbf{u} \in (\mathbb{F}^{*}_{q})^{r+1}$. The lemma is proved.
\end{proof}

The second lemma is given as follows.
\begin{lemma}\label{lem2.4}
\begin{description}
  \item[(i)] Suppose $(2s+1) \mid (q+1)$ and $m= \frac{q^{2}-1}{2s+1}$. Let $1 \leq k \leq (s+1+t)\frac{q+1}{2s+1}-1$, where $0\leq t \leq s-1$. Then for any $0 \leq i, j \leq k-1$, $m \mid (qi+j)$ if and only if $qi+j \in \{0, (s-t+1)m, (s-t+2)m, \ldots, (s+t)m\}$.
  \item[(ii)] Suppose $2s \mid (q+1)$ and $m= \frac{q^{2}-1}{2s}$. Let $1 \leq k \leq (s+1+t)\frac{q+1}{2s}-1$, where $0 \leq t \leq s-2$. Then for any $0 \leq i, j \leq k-1$, $m \mid (qi+j)$ if and only if $qi+j \in \{0, (s-t)m, (s-t+1)m, \ldots, (s+t)m\}$.
\end{description}
\end{lemma}

\begin{proof}
We only need to prove Part (i) since the proof of Part (ii) is completely similar.  According to the conditions, it is easy to see that $k \leq q-1$. Hence, for any $0 \leq i, j \leq k-1$ , we have $qi+j <(q+1)k \leq q^{2}-1$. Suppose $(i, j) \neq (0,0)$. If $qi+j=\ell m=\ell\frac{q^{2}-1}{2s+1}$, then $0< \ell < 2s+1$. Note that
\[ qi+j=\ell\frac{q^{2}-1}{2s+1}=q\left(\frac{\ell(q+1)}{2s+1}-1\right)+\left(q-\frac{\ell(q+1)}{2s+1}\right).\]
Thus
\[i=\frac{\ell(q+1)}{2s+1}-1, j=q-\frac{\ell(q+1)}{2s+1}.\]

If $\ell \geq s+1+t$, then
\[i=\frac{\ell(q+1)}{2s+1}-1 \geq (s+1+t)\frac{q+1}{2s+1}-1 \geq k,\]
which contradicts to the assumption that $i \leq k-1$;

If $\ell \leq s-t$,  then
\[j=q-\frac{\ell(q+1)}{2s+1}\geq (s+1+t)\frac{q+1}{2s+1}-1 \geq k,\]
which contradicts to the assumption that $j \leq k-1$.

Thus $s-t+1 \leq \ell \leq s+t$.
The conclusion follows.
\end{proof}

\section{Quantum MDS codes of length $n=1+r\frac{q^{2}-1}{s}$, where $s \mid (q-1)$}
In this section, we construct a class of quantum MDS codes of length $n=1+r\frac{q^{2}-1}{s}$, where $s \mid (q-1)$ and $1 \leq r \leq s$. We first prove the following lemma.

\begin{lemma}\label{lem3.1}
Let $x_{1}, \ldots, x_{r}$ be $r$ pairwise distinct nonzero elements of $\mathbb{F}_{q}$.
Then the system of equations
\begin{equation}\label{eq1}
 \left\{
\begin{aligned}
  u_{0}+ u_{1}+\cdots +u_{r}  = 0 \\
  x_{1}u_{1}+ x_{2}u_{2}+\cdots +x_{r}u_{r} = 0 \\
  \vdots ~~~~~~~~~~~~~~&  \\
   x_{1}^{r-1}u_{1}+ x_{2}^{r-1}u_{2}+\cdots +x_{r}^{r-1}u_{r}= 0
\end{aligned}\right.
\end{equation}
has a solution $\textbf{u}\triangleq(u_{0}, u_{1},\ldots, u_{r}) \in (\mathbb{F}_{q}^{*})^{r+1}$.
\end{lemma}
\begin{proof}
Let
\[A=\left(
      \begin{array}{ccccc}
        1 & 1 & 1 & \cdots &  1\\
        0 & x_{1} & x_{2} & \cdots & x_{r} \\
        \vdots & \vdots & \vdots & \ddots & \vdots \\
        0 & x_{1}^{r-1} & x_{2}^{r-1} & \cdots & x_{r}^{r-1}\\
      \end{array}
    \right)
.\]
Then the system (\ref{eq1}) of equations is equivalent to the following equation
\begin{equation*}
  A\textbf{u}^{T}=\textbf{0}^{T}.
\end{equation*}
Note that any $r$ columns of $A$ form a Vandermonde matrix, which is invertible. Thus any $r$ columns of $A$ are linearly independent. Since $x_{1}, \ldots, x_{r} \in \mathbb{F}_{q}$, $A^{(q)}=A$. The conclusion then follows from Lemma \ref{lem2.3}.
\end{proof}

\vskip 1mm

Set $m = \frac{q^{2}-1}{s}$. Let $\theta \in \mathbb{F}_{q^{2}}$ be an $m$-th primitive root of unity, and let $\langle\theta \rangle$ be the cyclic subgroup of $\mathbb{F}_{q^{2}}^{*}$ generated by $\theta$.  Let $\beta_{1}, \ldots , \beta_{r} \in \mathbb{F}_{q^{2}}^{*}$
such that $\{\beta_{i} \langle \theta \rangle\}^{r}_{i
=1}$ represent distinct cosets of $\mathbb{F}_{q^{2}}^{*}/\langle \theta \rangle$. Put
\[\textbf{a}=(0, \beta_{1}, \beta_{1}\theta, \ldots, \beta_{1}\theta^{m-1}, \ldots , \beta_{r}, \beta_{r}\theta, \ldots, \beta_{r}\theta^{m-1}) \in \mathbb{F}_{q^{2}}^{n}.\]
Set
\[\textbf{v}=(v_{0},\underbrace{ v_{1},\ldots,v_{1}}_{m\textnormal{ times}},\ldots,\underbrace{v_{r},\ldots,v_{r}}_{m\textnormal{ times}}),\]
where $v_{0}, v_{1},\ldots, v_{r} \in \mathbb{F}_{q^{2}}^{*}$.
Then
\begin{equation}\label{eq2}
  \langle\textbf{a}^{0}, \textbf{v}^{q+1}\rangle = v_{0}^{q+1}+ (v_{1}^{q+1}+\cdots+v_{r}^{q+1})m.
\end{equation}
And for any $(i,j)\neq (0,0)$, we have
\[\langle\textbf{a}^{qi+j}, \textbf{v}^{q+1}\rangle
=\sum_{\ell=1}^{r}\beta_{\ell}^{qi+j}v_{\ell}^{q+1}\sum_{\nu=0}^{m-1}\theta^{\nu(qi+j)},\]
thus
\[\langle\textbf{a}^{qi+j}, \textbf{v}^{q+1}\rangle=0, \textnormal{ when }m \nmid (qi+j),\]
and
\begin{equation}\label{eq3}
  \langle\textbf{a}^{qi+j}, \textbf{v}^{q+1}\rangle=m\sum\limits_{\ell=1}^{r}\beta_{\ell}^{qi+j}v_{\ell}^{q+1},  \textnormal{ when }m \mid (qi+j).
\end{equation}

Now, our first construction is given as follows.

\begin{theorem}\label{thm3.2}
Let $q$ be a prime power. Suppose $s \mid (q-1)$ and $1 \leq r \leq s$. Put $n=1+r\frac{q^{2}-1}{s}$. Then for any $1 \leq k \leq r\frac{q-1}{s}$, there exists an $[[n, n-2k, k+1]]_{q}$-quantum MDS code.
\end{theorem}
\begin{proof}
Keep the notations as above. Let $x_{\ell}=\beta_{\ell}^{m}$, for $\ell = 1, \ldots, r$. Then $x_{1}, \ldots, x_{r}$ are pairwise distinct. Indeed, if $x_{\ell} = x_{\ell'}$ for some $1 \leq \ell \neq \ell' \leq r$, then $(\frac{\beta_{\ell}}{\beta_{\ell'}})^{\frac{q^{2}-1}{s}}=1$ hence $\frac{\beta_{\ell}}{\beta_{\ell'}} \in \langle \theta \rangle$. This is impossible since $\beta_{\ell}$ and $\beta_{\ell'}$ lie in two distinct cosets of $\mathbb{F}_{q^{2}}^{*}/\langle \theta \rangle$.  Note that $(q+1) \mid m$, so $x_{\ell} \in \mathbb{F}_{q}$. Then, according to Lemma \ref{lem3.1}, there exists a vector $\textbf{u}=(u_{0}, u_{1}, \ldots, u_{r}) \in (\mathbb{F}_{q}^{*})^{r+1}$ which is a solution of the system (\ref{eq1}) of equations.

For $i=1, 2, \ldots, r$, we let $v_{i} \in \mathbb{F}_{q^{2}}^{*}$ such that $v_{i}^{q+1}=u_{i}$ and let $v_{0} \in \mathbb{F}_{q^{2}}^{*}$ such that $v_{0}^{q+1}=u_{0}m$. Then from Eq. (\ref{eq2}),
\begin{eqnarray*}
  \langle\textbf{a}^{0}, \textbf{v}^{q+1}\rangle &=& v_{0}^{q+1}+ (v_{1}^{q+1}+\cdots+v_{r}^{q+1})m \\
    &=& u_{0}m+(u_{1}+\cdots+u_{r})m=0.
\end{eqnarray*}
Since $1 \leq k \leq r\frac{q-1}{s}$, $qi+j \leq (q+1)(k-1)<r\frac{q^{2}-1}{s}=rm$. Thus, for any $0 \leq i, j\leq k-1$, $m \mid (qi+j)$ only if $qi+j=\mu m$ for some $0 \leq \mu \leq r-1$. Thus by Eq. (\ref{eq3}), when $qi+j=\mu m$ ($1 \leq \mu \leq r-1$),  we have
 \[\langle\textbf{a}^{qi+j}, \textbf{v}^{q+1}\rangle= m\sum_{\ell=1}^{r}\beta_{\ell}^{\mu m}v_{\ell}^{q+1}=m\sum_{\ell=1}^{r}x_{\ell}^{\mu}u_{\ell}=0.\]
In summary,
\[\langle \textbf{a}^{qi+j}, \textbf{v}^{q+1} \rangle =0,\textnormal{ for all } 0 \leq i, j \leq k-1.\]
By Lemma \ref{lem2.1}, $GRS_{k}(\textbf{a}, \textbf{v})$ is a Hermitian self-orthogonal MDS code with parameters $[n, k, n-k+1]$. The conclusion then follows from Corollary \ref{cor1.3}.
\end{proof}
 \begin{remark}
When $\frac{r}{s}>\frac{q}{2(q-1)}$, the quantum codes constructed in Theorem \ref{thm3.2} have minimum distance $r\frac{q-1}{s}+1 >\frac{q}{2}+1.$
\end{remark}

Applying the propagation rule (see Lemma \ref{lem1.4}) for Theorem \ref{thm3.2}, we immediately obtain the following corollary which is  one of main results in \cite{SYZ17}.

\begin{corollary} (\cite[Theorem 4.12]{SYZ17})\label{cor3.4}
Let $q$ be a prime power. Let $s \mid (q-1)$ and $1 \leq r \leq s$. Put $n=r\frac{q^{2}-1}{s}$. Then for any $1 \leq k \leq r\frac{q-1}{s}-1$, there exists an $[[n, n-2k, k+1]]_{q}$-quantum MDS code.
\end{corollary}

On the other hand, taking $r=s$ in Theorem \ref{thm3.2}, we obtain the following known result.

\begin{corollary} (\cite{GBR04})\label{cor3.5}
Let $q$ be a prime power. Then for any $1 \leq k \leq q-1$, there exists a $[[q^{2}, q^{2}-2k, k+1]]_{q}$-quantum MDS code.
\end{corollary}

In the following example, a new family of quantum MDS codes is given by Theorem \ref{thm3.2}.
\begin{example}
Let $(r, s) = (4, 5)$ in Theorem \ref{thm3.2}. Then when $5 \mid (q-1)$, there exists an $[[1+\frac{4}{5}(q^{2}-1), 1+\frac{4}{5}(q^{2}-1)-2k, k+1]]_{q}$ quantum MDS code for any $1 \leq k \leq \frac{4}{5}(q-1)$.
\end{example}

\section{Quantum MDS codes of length $n=1+r\frac{q^{2}-1}{2s+1}$, where $(2s+1) \mid (q+1)$}
In this section, we construct quantum MDS codes of length $n=1+r\frac{q^{2}-1}{2s+1}$, where $(2s+1) \mid (q+1)$. If $r=2s+1$, then $n=q^{2}$. The $q$-ary quantum MDS codes of length $q^{2}$ have been already constructed in \cite{GBR04} (see also Corollary \ref{cor3.5}). To simplify the following discussion, we assume that $1 \leq r < 2s+1$. Set $m=\frac{q^{2}-1}{2s+1}$. Before giving our construction, we need the following lemmas.

\begin{lemma}\label{lem4.1}
Suppose that $q > 2$ and $r \geq 1$. Then there exist $u_{0}, u_{1}, \ldots, u_{r} \in \mathbb{F}_{q}^{*}$ such that
 \[\sum_{i=0}^{r}u_{i}=0.\]
\end{lemma}
\begin{proof}
 We prove this lemma by induction on $r$. If $r=1$, this is trivial. For $r \geq 2$, by induction, the equation $\sum_{i=0}^{r}u_{i}=0$ has solutions $u_{0}, \ldots, u_{r-2}, (u_{r-1}+u_{r}):=u \in \mathbb{F}_{q}^{*}$. Now, take $u_{r-1} \in \mathbb{F}_{q}^{*}\backslash\{u\}$ and $u_{r}=u-u_{r-1} \neq 0$. The desired conclusion follows.
\end{proof}

\begin{lemma}\label{lem4.2}
Suppose $(2s+1) \mid (q+1)$ and $m=\frac{q^{2}-1}{2s+1}$. Let $\omega$ be a primitive element of $\mathbb{F}_{q^{2}}$ and $r=2t+1$, where $0 \leq t \leq s-1$. Then the following system of equations
\begin{equation}\label{eq4}
 \left\{
\begin{aligned}
   \sum_{\ell=0}^{r}u_{\ell} & =0 \\
  \sum_{\ell=1}^{r}\omega^{\ell \mu m}u_{\ell} &=0,\textnormal{ for }\mu=s-t+1, \ldots, s+t,
\end{aligned}\right.
\end{equation}
has a solution $\textbf{u}\triangleq(u_{0}, u_{1}, \ldots, u_{r}) \in (\mathbb{F}_{q}^{*})^{r+1}.$
\end{lemma}
\begin{proof}
 Let $\alpha=\omega^{m}$ be a primitive $(2s+1)$-th root of unity and let $a= s-t+1$. It is easy to verify that $\alpha^{a+\nu} \neq  \alpha^{a+\nu'} \neq 1$ for any $0 \leq \nu \neq \nu' \leq  r-2$. Let
\[A=\left(
      \begin{array}{ccccc}
       1& 1 & 1 & \cdots & 1 \\
       0 & \alpha^{a} & \alpha^{2a} & \cdots & \alpha^{ra} \\
       0 & \alpha^{a+1} & \alpha^{2(a+1)} & \cdots & \alpha^{r(a+1)} \\
       \vdots & \vdots & \vdots & \ddots & \vdots \\
       0 & \alpha^{a+r-2} & \alpha^{2(a+r-2)} & \cdots & \alpha^{r(a+r-2)}\\
      \end{array}
    \right)\]
be an $r \times (r+1)$ matrix over $\mathbb{F}_{q^{2}}$. Then the system (\ref{eq4}) of equations is equivalent to the following equation
\begin{equation*}
  A\textbf{u}^{T}=\textbf{0}^{T}.
\end{equation*}
For any $1 \leq i \leq r+1$, let $A_{i}$ be the $r\times r$ matrix obtained from $A$ by deleting the $i$-th column.
Then

\[\det(A_{1})=(\alpha^{(r-1)a+\frac{(r-1)(r-2)}{2}})\det (B_{1}) \neq 0,\]
where \[B_{1}=\left(
      \begin{array}{ccccc}
        1 & 1 & 1 & \cdots & 1 \\
        1 & \alpha^{a} & \alpha^{2a} & \cdots & \alpha^{(r-1)a} \\
        \vdots & \vdots & \vdots &\ddots & \vdots \\
        1 & \alpha^{a+r-2} & \alpha^{2(a+r-2)} & \cdots & \alpha^{(r-1)(a+r-2)} \\
      \end{array}
    \right),\]
and for $2 \leq i \leq r+1$
\[ \det(A_{i}) = b_{i} \det (B_{i})\neq 0,\]
where $b_{i}=\alpha^{a}\cdots\alpha^{(i-1)a}\alpha^{(i+1)a}\cdots\alpha^{ra}$ and
\[B_{i}=\begin{pmatrix}

        \begin{smallmatrix}

            1 & \cdots & 1 & 1 & \cdots & 1 \\
            \alpha & \cdots & \alpha^{i-1} & \alpha^{i+1} & \cdots & \alpha^{r} \\
             \alpha^{2} & \cdots & \alpha^{2(i-1)} & \alpha^{2(i+1)} & \cdots & \alpha^{2r} \\
             \vdots & \ddots & \vdots & \vdots & \ddots & \vdots \\
           \alpha^{r-2} & \cdots & \alpha^{(i-1)(r-2)} & \alpha^{(i+1)(r-2)} & \cdots & \alpha^{r(r-2)}

        \end{smallmatrix}

    \end{pmatrix}.\]
Hence any $r$ columns of $A$ are linearly independent.
On the other hand, since $(2s+1) \mid (q+1)$, we have
\[\alpha^{i(a+j)q}=\alpha^{-i(s-t+1+j)}=\alpha^{i(s+t-j)}=\alpha^{i(a+r-2-j)},\]
for any $1 \leq i \leq r$ and $0 \leq j \leq r-2$.
Thus $A$ is row equivalent to $A^{(q)}$.  The conclusion then follows from Lemma \ref{lem2.3}.

\end{proof}

Let $\omega$ be a primitive element of $\mathbb{F}_{q^{2}}$ and $\theta=\omega^{2s+1}$ be a primitive $m$-th root of unity ($m=\frac{q^{2}-1}{2s+1}$). It is easy to verify that
\[\omega^{i_{1}}\theta^{j_{1}} \neq \omega^{i_{2}}\theta^{j_{2}} \]
for any $1 \leq i_{1} \neq i_{2} \leq r$ and $0 \leq j_{1} \neq j_{2} \leq m-1.$
Put
\[\textbf{a}=(0, \omega, \omega\theta, \ldots, \omega\theta^{m-1}, \ldots , \omega^{r}, \omega^{r}\theta, \ldots, \omega^{r}\theta^{m-1}) \in \mathbb{F}_{q^{2}}^{n}.\]
Set
\[\textbf{v}=(v_{0},\underbrace{ v_{1},\ldots,v_{1}}_{m\textnormal{ times}},\ldots,\underbrace{v_{r},\ldots,v_{r}}_{m\textnormal{ times}}),\]
where $v_{0}, v_{1},\ldots, v_{r} \in \mathbb{F}_{q^{2}}^{*}$.
Similar to the discussion before Theorem \ref{thm3.2}, we have
\begin{equation}\label{eq5}
  \langle\textbf{a}^{0}, \textbf{v}^{q+1}\rangle = v_{0}^{q+1}+ (v_{1}^{q+1}+\cdots+v_{r}^{q+1})m.
\end{equation}
For any $(i,j)\neq (0,0)$,
\[\langle\textbf{a}^{qi+j}, \textbf{v}^{q+1}\rangle=0, \textnormal{ when }m \nmid (qi+j),\]
and
\begin{equation}\label{eq6}
  \langle\textbf{a}^{qi+j}, \textbf{v}^{q+1}\rangle=m\sum\limits_{\ell=1}^{r}\omega^{\ell(qi+j)}v_{\ell}^{q+1},  \textnormal{ when }m \mid (qi+j).
\end{equation}

Now, we present our second construction as follows.
\begin{theorem}\label{thm4.3}
Let $q > 2$ be a prime power, $(2s+1) \mid (q+1)$ and $1 \leq r < 2s+1$. Put $n=1+r\frac{q^{2}-1}{2s+1}$.

\begin{description}
  \item[(i)] For any $1 \leq k \leq (s+1)\frac{q+1}{2s+1}-1$, there exists an $[[n, n-2k, k+1]]_{q}$-quantum MDS code.
  \item[(ii)] If $r=2t+1$, where $0 \leq t \leq s-1$, then for any $1 \leq k \leq (s+1+t)\frac{q+1}{2s+1}-1$, there exists an $[[n, n-2k, k+1]]_{q}$-quantum MDS code.
\end{description}
\end{theorem}

\begin{proof}
Keep the notations as above.
\vskip 1mm
(i): Suppose $1 \leq k \leq (s+1)\frac{q+1}{2s+1}-1$.
By Lemma \ref{lem4.1}, there exist $u_{0}, u_{1}, \ldots, u_{r} \in \mathbb{F}_{q}^{*}$ such that
 \[\sum_{i=0}^{r}u_{i}=0.\]
For $i=1, 2, \ldots, r$, let $v_{i} \in \mathbb{F}_{q^{2}}^{*}$ such that $v_{i}^{q+1}=u_{i}$ and let $v_{0} \in \mathbb{F}_{q^{2}}^{*}$ such that $v_{0}^{q+1}=u_{0}m$.
Then by Eq. (\ref{eq5}),
\[\langle\textbf{a}^{0}, \textbf{v}^{q+1}\rangle = u_{0}m+(u_{1}+\cdots+u_{r})m=0.\]
Taking $t=0$ in Lemma \ref{lem2.4} (i), we obtain that $m \mid (qi+j)$ if and only if $(i,j)=(0,0)$. Thus from the above discussion,
\[\langle \textbf{a}^{qi+j}, \textbf{v}^{q+1} \rangle =0,\textnormal{ for all } 0 \leq i, j \leq k-1.\]
By Lemma \ref{lem2.1}, $GRS_{k}(\textbf{a}, \textbf{v})$ is a Hermitian self-orthogonal MDS code with parameters $[n, k, n-k+1]$. Part (i) then follows from Corollary \ref{cor1.3}.
\vskip 2mm
(ii): Suppose $r=2t+1$, where $0 \leq t \leq s-1$ and $1 \leq k \leq (s+t+1)\frac{q+1}{2s+1}-1$. By Lemma \ref{lem4.2}, there exist $u_{0}, u_{1}, \ldots, u_{r} \in \mathbb{F}_{q}^{*}$ which satisfy the system (\ref{eq4}) of equations. For $i=1, 2, \ldots, r$, let $v_{i} \in \mathbb{F}_{q^{2}}^{*}$ such that $v_{i}^{q+1}=u_{i}$ and let $v_{0} \in \mathbb{F}_{q^{2}}^{*}$ such that $v_{0}^{q+1}=u_{0}m$.
Then by Eq. (\ref{eq6}),
\[\langle\textbf{a}^{0}, \textbf{v}^{q+1}\rangle = u_{0}m+(u_{1}+\cdots+u_{r})m=0.\]
 By Lemma \ref{lem2.4} (i), $m \mid (qi+j)$  if and only if $qi+j \in \{0, (s-t+1)m, (s-t+2)m, \ldots, (s+t)m \}$. Thus by Eq. (\ref{eq6}), when $qi+j=\mu m$ ($s-t+1 \leq \mu \leq s+t$),  we have
 \[\langle\textbf{a}^{qi+j}, \textbf{v}^{q+1}\rangle= m\sum_{\ell=1}^{r}\omega^{\ell\mu m}v_{\ell}^{q+1}=m\sum_{\ell=1}^{r}\omega^{\ell\mu m}u_{\ell}=0.\]
Hence
\[\langle\textbf{a}^{qi+j}, \textbf{v}^{q+1}\rangle=0,\]
for all $ 0 \leq i, j \leq k-1.$ By Lemma \ref{lem2.1}, $GRS_{k}(\textbf{a}, \textbf{v})$ is a Hermitian self-orthogonal MDS code with parameters $[n, k, n-k+1]$. Part (ii) then also follows from Corollary \ref{cor1.3}.
\\The proof of this theorem is completed.
\end{proof}
\begin{remark}\label{rem4.4}
\begin{description}
  \item[i)] The minimum distance of the quantum codes constructed in Theorem \ref{thm4.3} can be larger than $\frac{q}{2}+1$.
  \item[ii)]  Part (i) of Theorem \ref{thm4.3} extends the result of \cite[Theorem 3.2 (i)]{JKW17}  where a stricter condition $\textnormal{gcd}(r,q) = 1$  is required.
  \item[iii)] When $r=2t+1$ and $1 \leq t \leq s-1$, the quantum codes from Part (ii) of Theorem \ref{thm4.3} have larger minimum distance than that of Part (i).
\end{description}
\end{remark}

Shi \emph{et al.} \cite[Theorem 4.2]{SYZ17} constructed a family of quantum MDS codes of length $n=r\frac{q^{2}-1}{2s+1}$, where $r=2t+2$ is even. For $r=2t+1$ odd, applying the propagation rule (see Lemma \ref{lem1.4}) for Theorem \ref{thm4.3} (ii), we can immediately obtain the following result.
\begin{corollary}\label{cor4.5}
Let $q > 2$ be a prime power, $(2s+1) \mid (q+1)$ and $0 \leq t \leq s-1$. Put $n=(2t+1)\frac{q^{2}-1}{2s+1}$. Then for any $1 \leq k \leq (s+1+t)\frac{q+1}{2s+1}-2$, there exists an $[[n, n-2k, k+1]]_{q}$-quantum MDS code.
\end{corollary}

\begin{remark}\label{rem4.6}
Jin \emph{et al.} \cite[Theorem 3.2 (ii)]{JKW17} constructed a family of $q$-ary quantum MDS codes with parameters $[[r\frac{q^{2}-1}{2s+1}, r\frac{q^{2}-1}{2s+1}-2k, k+1]]$, for any $k \leq (s+1)\frac{q+1}{2s+1}-1$, where $(2s+1) \mid (q+1)$ and $\gcd(r, q) >1$. If $t \geq 1$ and $2s+1 \neq q+1$, then $(s+1+t)\frac{q+1}{2s+1}-1 \geq (s+1)\frac{q+1}{2s+1}$ and hence the quantum codes of Corollary \ref{cor4.5} have larger minimum distance.
\end{remark}

\begin{example}In this example, we give some new quantum MDS codes from Theorem \ref{thm4.3}.
\begin{description}
  \item[(i)] Let $(r, s)=(2,2)$ in Theorem \ref{thm4.3} (i). Then, when $5 \mid (q+1)$, there exists a $[[1+\frac{2}{5}(q^{2}-1), 1+\frac{2}{5}(q^{2}-1)-2k, k+1]]_{q}$ quantum MDS code for any $1 \leq k \leq \frac{3}{5}(q+1)-1$;
  \item[(ii)] Let $(r,s)=(5, 3)$ in Theorem \ref{thm4.3} (ii). Then, when $7 \mid (q+1)$, there exists a $[[1+\frac{5}{7}(q^{2}-1), 1+\frac{5}{7}(q^{2}-1)-2k, k+1]]_{q}$ quantum MDS code for any $1 \leq k \leq \frac{6}{7}(q+1)-1$.
\end{description}
\end{example}

\section{Quantum MDS codes of length $n=1+r\frac{q^{2}-1}{2s}$, where $2s \mid (q+1)$}
In this section, we construct  quantum MDS codes of length $n=1+r\frac{q^{2}-1}{2s}$, where $1 \leq r \leq 2s$ and $2s \mid (q+1)$. If $r=2s$, then $n=q^{2}$; If $r=s=1$, then $n=\frac{q^{2}+1}{2}$. The $q$-ary quantum MDS codes of lengths $q^{2}$ and $\frac{q^{2}+1}{2}$ have been already constructed in \cite{GBR04} and \cite{KZ12}, respectively. To simplify the following discussion, we assume that $r < 2s$ and $s > 1$. In this section, we denote $m:=\frac{q^{2}-1}{2s}$. We first provide two technical lemmas as follows.
\begin{lemma}\label{lem5.1}
Suppose that $q$ is odd and $r \geq 2$. Then the following system of equations
\begin{equation}\label{eq7}
\left\{
\begin{aligned}
 \sum_{k=0}^{r}u_{k} & =0 \\ \sum_{i=1}^{r}(-1)^{i}u_{i} & =0
\end{aligned}\right.
\end{equation}
has a solution $\textbf{u}\triangleq(u_{0}, u_{1}, \ldots, u_{r}) \in (\mathbb{F}_{q}^{*})^{r+1}$.
\end{lemma}
\begin{proof}
Note that the system (\ref{eq7}) of equations is equivalent to
\[u_{0}+\sum_{i=1, i\textnormal{ odd}}^{r}(2u_{i})=u_{0}+\sum_{j=2, j \textnormal{ even}}^{r}(2u_{j})=0.\]
The conclusion then follows from Lemma \ref{lem4.1}.
\end{proof}

According to Lemma \ref{lem2.3}, we can prove the following lemma similarly as Lemma \ref{lem4.2}. Hence, we omit the details of the proof.
\begin{lemma}\label{lem5.2}
Suppose $2s \mid (q+1)$. Let $\omega$ be a primitive element of $\mathbb{F}_{q^{2}}$ and $r=2t+2$, where $0 \leq t \leq s-2$. Then the following system of equations
\begin{equation*}\label{7}
 \left\{
\begin{aligned}
  \sum_{\ell=0}^{r}u_{\ell} &=0 \\
  \sum_{\ell=1}^{r}\omega^{\ell \mu m}u_{\ell} &=0,~~\mu=s-t, s-t+1, \ldots, s+t,
\end{aligned}\right.
\end{equation*}
has a solution $\textbf{u}\triangleq(u_{0}, u_{1}, \ldots, u_{r}) \in (\mathbb{F}_{q}^{*})^{r+1}.$
\end{lemma}

Let $\omega$ be a primitive element of $\mathbb{F}_{q^{2}}$ and $\theta=\omega^{2s}$ be a primitive $m$-th root of unity.
Put
\[\textbf{a}=(0, \omega, \omega\theta, \ldots, \omega\theta^{m-1}, \ldots , \omega^{r}, \omega^{r}\theta, \ldots, \omega^{r}\theta^{m-1}) \in \mathbb{F}_{q^{2}}^{n}.\]
Now, we give our third construction as follows.
\begin{theorem}\label{thm5.3}
Let $q$ be a prime power, $2s \mid (q+1)$ and $2 \leq r < 2s$. Put $n=1+r\frac{q^{2}-1}{2s}$.

\begin{description}
  \item[(i)] For any $1 \leq k \leq (s+1)\frac{q+1}{2s}-1$, there exists an $[[n, n-2k, k+1]]_{q}$-quantum MDS code.
  \item[(ii)] If $r=2t+2$, where $0 \leq t \leq s-2$, then for any $1 \leq k \leq (s+t+1)\frac{q+1}{2s}-1$, there exists an $[[n, n-2k, k+1]]_{q}$-quantum MDS code.
\end{description}
\end{theorem}

\begin{proof}
By employing Lemmas \ref{lem2.1}, \ref{lem4.2} and \ref{lem5.2}, the theorem can be proved similarly as Theorem \ref{thm4.3}. We omit the details.
\end{proof}
\begin{remark}\label{rem5.4}
\begin{description}
  \item[i)] The minimum distance of quantum codes constructed in Theorem \ref{thm5.3} can be larger than $\frac{q}{2}+1$.
  \item[ii)] When $r=2t+2$ and $1 \leq t \leq s-2$, the quantum codes from Part (ii) of Theorem \ref{thm5.3} have larger minimum distance than that of Part (i).
\end{description}
\end{remark}

Applying the propagation rule (see Lemma \ref{lem1.4}) for Theorem \ref{thm5.3} (i) and (ii), we immediately obtain the following corollaries which were given in \cite{ZG17} and \cite{SYZ17}, respectively.

\begin{corollary} (\cite[Theorem 4.2]{ZG17})\label{cor5.5}
Let $q$ be a prime power, $2s\mid (q+1)$ and $2 \leq r < 2s$. Put $n=r\frac{q^{2}-1}{2s}$. Then for any $1 \leq k \leq (s+1)\frac{q+1}{2s}-2$, there exists an $[[n, n-2k, k+1]]_{q}$-quantum MDS code.
\end{corollary}

\begin{corollary} (\cite[Theorem 4.8]{SYZ17})\label{cor5.6}
Let $q$ be a prime power, $2s\mid (q+1)$ and $0 \leq t \leq s-2$. Put $n=(2t+2)\frac{q^{2}-1}{2s}$. Then for any $1 \leq k \leq (s+t+1)\frac{q+1}{2s}-2$, there exists an $[[n, n-2k, k+1]]_{q}$-quantum MDS code.
\end{corollary}

\begin{example}In this example, we give some new quantum MDS codes from Theorem \ref{thm5.3}.
\begin{description}
  \item[(i)] Let $(r, s)=(3,2)$ in Theorem \ref{thm5.3} (i). Then, when $4 \mid (q+1)$, there exists a $[[1+\frac{3}{4}(q^{2}-1), 1+\frac{3}{4}(q^{2}-1)-2k, k+1]]_{q}$ quantum MDS code for any $1 \leq k \leq \frac{3}{4}(q+1)-1$;
  \item[(ii)] Let $(r,s)=(4, 3)$ in Theorem \ref{thm5.3} (ii). Then, when $6 \mid (q+1)$, there exists a $[[1+\frac{2}{3}(q^{2}-1), 1+\frac{2}{3}(q^{2}-1)-2k, k+1]]_{q}$ quantum MDS code for any $1 \leq k \leq \frac{5}{6}(q+1)-1$.
\end{description}
\end{example}

\section{Quantum MDS codes of length $n=(2t+1)\frac{q^{2}-1}{2s}$, where $2s \mid (q+1)$}

Suppose $2s \mid (q+1)$ and $0 \leq t \leq s-1$. In \cite[Theorem 4.8]{SYZ17}, Shi \emph{et al.} constructed a family of quantum MDS codes of length $(2t+2)\frac{q^{2}-1}{2s}$ (see Corollary \ref{cor5.6}). In this section, we contribute to construct a family of quantum MDS codes of length $(2t+1)\frac{q^{2}-1}{2s}$. Before giving our construction, we need the following lemmas.

\begin{lemma}\label{lem6.1}
Suppose that $3 \leq \tau < q+1$. Let $M$ be a $(\tau-2) \times \tau$  matrix over $\mathbb{F}_{q^{2}}$ and satisfy the following two properties: 1) $M$ and $M^{(q)}$ are row equivalent; 2) any $\tau-2$ columns of $M$ are linearly independent. Then the following equation
\begin{equation*}\label{9}
  M\textbf{x}^{T}=\textbf{0}^{T}
\end{equation*}
has a solution $\textbf{x}=(x_{1}, x_{2}, \ldots, x_{\tau}) \in (\mathbb{F}_{q}^{*})^{\tau}$.
\end{lemma}
\begin{proof}
  Let $M_{1}$ (resp. $M_{\tau}$) be the $(\tau-2)\times (\tau-1)$ matrix obtained from $M$ by deleting the first (resp. the last) column. From the conditions, we obtain that $M_{1}$ and $M_{\tau}$ satisfy the properties in Lemma \ref{lem2.3} (let $r=\tau-1$). Thus the following two equations
    \begin{equation*}\label{10}
      M_{1}\textbf{u}^{T}=\textbf{0}^{T},~ M_{\tau}\textbf{v}^{T}=\textbf{0}^{T}
\end{equation*}
have nonzero solutions $\textbf{u}=(u_{2}, u_{3}, \ldots, u_{\tau}) \in (\mathbb{F}_{q}^{*})^{\tau-1}$ and $\textbf{v}=(v_{1}, v_{2}, \ldots, v_{\tau-1}) \in (\mathbb{F}_{q}^{*})^{\tau-1},$ respectively. Since $\tau < q+1$, we may choose an element $\alpha \in \mathbb{F}_{q}^{*}\backslash \{\frac{u_{2}}{v_{2}},\ldots, \frac{u_{\tau-1}}{v_{\tau-1}}\}$. Let $\textbf{x}=(0, \textbf{u})-\alpha(\textbf{v},0)$, then $\textbf{x} \in (\mathbb{F}_{q}^{*})^{\tau}$ and
\[M\textbf{x}^{T}=\left(
\begin{array}{c}
      0 \\
      M_{1}\textbf{u}^{T} \\
\end{array}
\right)
+\left(
\begin{array}{c}
      M_{\tau}\textbf{v}^{T} \\
      0 \\
\end{array}
\right)=\textbf{0}^{T}.\]
The lemma is proved.
\end{proof}
\begin{lemma}\label{lem6.2}
 Suppose $2s \mid (q+1)$ and $m=\frac{q^{2}-1}{2s}$. Let $\omega$ be a primitive element of $\mathbb{F}_{q^{2}}$ and $r=2t+1$, where $1 \leq t \leq s-1$.
Then the following system of equations
\begin{equation}\label{eq8}
\sum_{\ell=1}^{r}\omega^{\ell (\mu m-q-1)}u_{\ell} =0,\textnormal{ for }\mu=s-t+1,  \ldots, s+t-1,
\end{equation}
has a solution $\textbf{u}\triangleq(u_{1}, u_{2}, \ldots, u_{r}) \in (\mathbb{F}_{q}^{*})^{r}.$
\end{lemma}
\begin{proof}
Denote $\alpha=\omega^{m}$, $\eta=\omega^{-q-1}$ and $a=s-t+1$. Then $\alpha^{2s}=1$ and $\eta \in \mathbb{F}_{q}$. Let
\[M=\left(
      \begin{array}{cccc}
        \alpha^{a}\eta & \alpha^{2a}\eta^{2} & \cdots & \alpha^{ra}\eta^{r} \\
        \alpha^{a+1}\eta & \alpha^{2(a+1)}\eta^{2} & \cdots & \alpha^{r(a+1)}\eta^{r} \\
        \vdots & \vdots & \ddots & \vdots \\
        \alpha^{a+r-3}\eta & \alpha^{2(a+r-3)}\eta^{2} & \cdots & \alpha^{r(a+r-3)}\eta^{r}\\
      \end{array}
    \right)\]
be an $(r-2) \times r$ matrix over $\mathbb{F}_{q^{2}}$.
Then the system (\ref{eq8}) of equations is equivalent to the following equation
\begin{equation}\label{eq9}
  M\textbf{u}^{T}=\textbf{0}^{T}.
\end{equation}
Since $2s \mid (q+1)$, we have
\begin{eqnarray*}
  (\alpha^{i(a+j)}\eta^{i})^{q} &=& \alpha^{qi(a+j)}\eta^{qi}=\alpha^{-i(s-t+1+j)}\eta^{i} \\
   &=& \alpha^{i(s+t-1-j)}\eta^{i}=\alpha^{i(a+r-3-j)}\eta^{i},
\end{eqnarray*}
for any $1 \leq i \leq r$ and $0 \leq j \leq r-3$.
Thus $M$ is row equivalent to $M^{(q)}$.  Let $M_{ij}$ $(1 \leq i \neq j \leq r)$ be the $(r-2)\times(r-2)$ matrix obtained from $M$ by deleting the $i$-th and $j$-th columns. It is not hard to verify that $\det(M_{ij}) \neq 0$. Thus by Lemma \ref{lem6.1}, Eq. (\ref{eq9}) has a  solution $\textbf{u} \in (\mathbb{F}_{q}^{*})^{r}$.
The lemma is proved.

\end{proof}

Set $m=\frac{q^{2}-1}{2s}$. Let $\omega$ be a primitive element of $\mathbb{F}_{q^{2}}$ and $\theta=\omega^{2s}$ be a primitive $m$-th root of unity. Put
\[\textbf{a}=(\omega, \omega\theta, \ldots, \omega\theta^{m-1}, \ldots , \omega^{r}, \omega^{r}\theta, \ldots, \omega^{r}\theta^{m-1}) \in \mathbb{F}_{q^{2}}^{n}.\]
Set
\[\textbf{v}=( v_{1},v_{1}\theta, \ldots,v_{1}\theta^{m-1},\ldots, v_{r}, v_{r}\theta,\ldots,v_{r}\theta^{m-1}),\]
where $v_{1}, \ldots, v_{r} \in \mathbb{F}_{q^{2}}^{*}$.
Then
\[\langle\textbf{a}^{qi+j}, \textbf{v}^{q+1}\rangle
=\sum_{\ell=1}^{r}\omega^{\ell(qi+j)}v_{\ell}^{q+1}\sum_{\nu=0}^{m-1}\theta^{\nu(qi+j+q+1)}.\]
Thus
\[\langle\textbf{a}^{qi+j}, \textbf{v}^{q+1}\rangle=0, \textnormal{ when } m \nmid (qi+j+q+1),\]
and
\begin{equation}\label{eq10}
  \langle\textbf{a}^{qi+j}, \textbf{v}^{q+1}\rangle=m\sum_{\ell=1}^{r}\omega^{\ell(qi+j)}v_{\ell}^{q+1}, \textnormal{ when }m \mid (qi+j+q+1).
\end{equation}

Now, we give our last construction as follows.
\begin{theorem}\label{thm6.3}
Let $q$ be a prime power. Suppose $2s \mid (q+1)$ and $r=2t+1$, where $1 \leq t \leq s-1$. Put $n=r\frac{q^{2}-1}{2s}$.
 Then for any $1 \leq k \leq (s+t)\frac{q+1}{2s}-2$, there exists an $[[n, n-2k, k+1]]_{q}$-quantum MDS code.
\end{theorem}

\begin{proof}
Keep the notations as above. By Lemma \ref{lem6.2}, there exist $u_{1}, \ldots, u_{r} \in \mathbb{F}_{q}^{*}$ such that
 \[\sum_{\ell=1}^{r}\omega^{\ell (\mu m-q-1)}u_{\ell} =0,\]
 for all $s-t+1 \leq \mu \leq s+t-1$.
For $1 \leq i \leq r$, we let $v_{i} \in \mathbb{F}_{q^{2}}^{*}$ such that $v_{i}^{q+1}=u_{i}$. Note that $qi+j+q+1=q(i+1)+(j+1)$. We can prove similarly as Lemma \ref{lem2.4} (ii) that $m \mid (qi+j+q+1)$ if and only if $qi+j+q+1 \in \{(s-t+1)m, (s-t+2)m, \ldots, (s+t-1)m \}$. Hence from Eq. (\ref{eq10}), when $qi+j+q+1=\mu m$ ($s-t+1 \leq \mu \leq  s+t-1$), we have
\begin{eqnarray*}
  \langle\textbf{a}^{qi+j}, \textbf{v}^{q+1}\rangle &=& m\sum_{\ell=1}^{r}\omega^{\ell(\mu m-q-1)}v_{\ell}^{q+1} \\
    &=& m\sum_{\ell=1}^{r}\omega^{\ell(\mu m-q-1)}u_{\ell}=0.
\end{eqnarray*}
Thus
\[\langle \textbf{a}^{qi+j}, \textbf{v}^{q+1} \rangle =0,\textnormal{ for all } 0 \leq i, j \leq k-1.\]
By Lemma \ref{lem2.1}, $GRS_{k}(\textbf{a}, \textbf{v})$ is a Hermitian self-orthogonal MDS code with parameters $[n, k, n-k+1]$. Theorem \ref{thm6.3} then follows from Corollary \ref{cor1.3}.
\end{proof}

According to Theorem \ref{thm6.3} and Corollary \ref{cor5.6}, we obtain the following corollary.
\begin{corollary}\label{cor6.4}
Let $q$ be a prime power. Suppose $2s \mid (q+1)$ and  $3 \leq r \leq 2s$. Put $n=r\frac{q^{2}-1}{2s}$.
 Then for any $1 \leq k \leq (s+\lceil\frac{r-1}{2}\rceil)\frac{q+1}{2s}-2$, there exists an $[[n, n-2k, k+1]]_{q}$-quantum MDS code.
\end{corollary}
\begin{remark}\label{rem6.5}
Zhang and Ge \cite[Theorem 4.2]{ZG17} (see also Corollary \ref{cor5.5}) constructed a family of $q$-ary quantum MDS codes with parameters $[[r\frac{(q^{2}-1)}{2s},r\frac{(q^{2}-1)}{2s}-2k, k+1]]$, $k \leq (s+1)\frac{q+1}{2s+1}-2$, where $2s \mid (q+1)$. If $r \geq 4$, then $(s+\lceil\frac{r-1}{2}\rceil)\frac{q+1}{2s}-1 > (s+1)\frac{q+1}{2s}-1$ and hence the quantum codes of Corollary \ref{cor6.4} have larger minimum distance.
\end{remark}

In the following example, a new family of quantum MDS codes is given from Theorem \ref{thm6.3}.
\begin{example}
Let $(r, s)=(7,4)$ in Theorem \ref{thm6.3}. Then, when $8 \mid (q+1)$, there exists a $[[\frac{7}{8}(q^{2}-1), \frac{7}{8}(q^{2}-1)-2k, k+1]]_{q}$ quantum MDS code for any $1 \leq k \leq \frac{7}{8}(q+1)-2$.
\end{example}
\section{Conclusion}
In this paper, we have constructed six new classes of $q$-ary quantum MDS codes by using Hermitian self-orthogonal GRS codes. Most of our quantum MDS codes have minimum distance larger than $\frac{q}{2}+1$. Some quantum MDS codes presented in \cite{ZG17} and \cite{SYZ17} can be easily derived from ours via the propagation rule. We also generalize and improve some results in \cite{ZG17},\cite{SYZ17}, and \cite{JKW17}.

\section*{Acknowledgment}
This work was supported in part by the 973 Program of
China under Grant 2013CB834204, in part by the National Natural Science
Foundation of China under Grant 61571243 and Grant 61771273, in part by
the Nankai Zhide Foundation, and in part by the Fundamental Research Funds
for the Central Universities of China.

\end{document}